\documentclass[twocolumn,18pt]{IEEEtran}
\usepackage{algorithm,algorithmic,amsbsy,amsmath,amssymb,epsfig,bbm,mathrsfs,multirow,amsthm}
\usepackage{array,multirow,graphicx}
\usepackage{graphicx}
\usepackage{graphics}
\usepackage{subfigure}
\usepackage{epstopdf}
\usepackage{color}
\usepackage{bbm}
\usepackage{comment}

\newcommand{\beq}{\begin{equation}}
\newcommand{\eeq}{\end{equation}}
\newcommand{\bitm}{\begin{itemize}}
\newcommand{\ba}{\begin{array}}
\newcommand{\ea}{\end{array}}
\newcommand{\eitm}{\end{itemize}}
\newcommand{\beqn}{\begin{eqnarray}}
\newcommand{\eeqn}{\end{eqnarray}}
\newcommand{\beqno}{\begin{eqnarray*}}
\newcommand{\eeqno}{\end{eqnarray*}}
\newcommand{\bma}{\begin{displaymath}}
\newcommand{\ema}{\end{displaymath}}
\newcommand{\bnu}{\begin{enumerate}}
\newcommand{\enu}{\end{enumerate}}
\newcommand{\bce}{\begin{center}}
\newcommand{\ece}{\end{center}}
\newcommand{\btb}{\begin{tabular}}
\newcommand{\etb}{\end{tabular}}

\newtheorem{Definition}{\hskip 0pt Definition}%[section]
\newtheorem{lemma}{Lemma}

\newtheorem{theorem}{\textbf{\textsc{Theorem}}}

\begin{document}
\title{Evolutionary Game for Mining Pool Selection in Blockchain Networks}

\author{
\IEEEauthorblockN{Xiaojun Liu\IEEEauthorrefmark{1}\IEEEauthorrefmark{2},
Wenbo Wang\IEEEauthorrefmark{2},
Dusit Niyato\IEEEauthorrefmark{2},
Narisa Zhao\IEEEauthorrefmark{1}
and
Ping Wang\IEEEauthorrefmark{2}}

\IEEEauthorblockA{\IEEEauthorrefmark{1}Institute of Systems Engineering, Dalian University of Technology, Dalian, China, 116024}

\IEEEauthorblockA{\IEEEauthorrefmark{2}School of Computer Engineering, Nanyang Technological University, Singapore, 639798}}
\maketitle
\begin{abstract}
  In blockchain networks adopting the proof-of-work schemes, the monetary incentive is introduced by the Nakamoto consensus protocol to guide the behaviors of the full nodes (i.e., block miners) in the process of maintaining the consensus about the blockchain state. The block miners have to devote their computation power measured in hash rate in a crypto-puzzle solving competition to win the reward of publishing (a.k.a., mining) new blocks. Due to the exponentially increasing difficulty of the crypto-puzzle, individual block miners tends to join mining pools, i.e., the coalitions of miners, in order to reduce the income variance and earn stable profits. In this paper, we study the dynamics of mining pool selection in a blockchain network, where mining pools may choose arbitrary block mining strategies. We identify the hash rate and the block propagation delay as two major factors determining the outcomes of mining competition, and then model the strategy evolution of the individual miners as an evolutionary game. We provide the theoretical analysis of the evolutionary stability for the pool selection dynamics in a case study of two mining pools. The numerical simulations provide the evidence to support our theoretical discoveries as well as demonstrating the stability in the evolution of miners' strategies in a general case.
\end{abstract}

\section{Introduction}\label{Sec:introduction}
Since its introduction in the grassroot online project ``Bitcoin''~\cite{nakamoto2008bitcoin}, the technology of blockchain has attracted significant attentions across the academia, the industry and the public. A blockchain network is built upon a virtual overlay Peer-to-Peer (P2P) network as a decentralized temper-proof system for transactional record logging~\cite{7423672}. For permissionless blockchains, the Nakamoto consensus protocol~\cite{nakamoto2008bitcoin} is widely adopted to financially incentivize the full nodes (block miners) to abide by the “longest-chain rule” in maintaining the blockchain state. In a blockchain network, the blockchain users issue the digitally signed transactions between their cryptographic addresses (a.k.a., wallets). Following the Nakamoto protocol, the block miners pack an arbitrary number of such verified transactions into a data structure, known as a candidate block, and broadcast it to the rest of the network. A blockchain is thus organized as a hash-linked list of such blocks and stored distributively as the local replica on each block miner. Following the ``longest-chain rule'', at a given time instance, the longest one among the proposed blockchain views will be ultimately recognized by the network as the current state of the blockchain~\cite{7423672}. Both the engineering practices and theoretical studies~\cite{Garay2015} have shown that the Nakamoto protocol is able to guarantee the persistence and liveness of a blockchain in a Byzantine environment. In other words, when a majority of the mining nodes honestly follow the Nakamoto protocol, the transactional data on the blockchain are guaranteed to be immutable once they are recorded.

In the Nakamoto protocol, the financial incentive mechanism consists of two parts: (a) a computation-intensive crypto-puzzle solving process to make Sybil attacks financially unaffordable, and (b) a reward generation process to award the miners when their published blocks are recognized by the entire network. The crypto-puzzle solving process is implemented through a Proof-of-Work (PoW) process~\cite{Garay2015}, where the miners exhaustively query a trusted random oracle (e.g., a SHA-256 hash function) to find a string satisfying a preimage condition based on their individual block proposals. In the block awarding process, the first miner whose candidate block gets accepted by the majority of the network will receive a monetary reward for its contribution in approving transactions and extending the blockchain by one block. Except for the transaction fees named by the transaction issuers, the winner in a round of the block mining competitions may also receive an \emph{ex-nihilo}, fixed-amount of award according to the token-generation mechanism of the blockchain~\cite{nakamoto2008bitcoin,7423672}.

The probability of winning a puzzle-solving competition depends on the ratio between a miner's hash rate (i.e., number of queries to the random oracle per unit time) and the total hash rate of the entire network~\cite{Garay2015}. In addition, the block propagation time in the P2P network determines the final result of block confirmation within one consensus round, since only the first block propagated to the majority of the network can be confirmed as the new head of the blockchain~\cite{Rizun2015, Houy2016}. In practical scenarios, the chance for individual miners to win in one round of the mining competitions has been negligible due to the exponential growth of hash rate in the network. As a result, the real-world blockchain networks are dominated by the proxy nodes that represent the coalitions of miners known as mining pools~\cite{7423672}. A mining pool works as a task scheduler by dividing a puzzle-solution task into smaller sub-tasks and assigning them to the miners according to their devoted hash rate. By aggregating the hash rate of many miners, the probability for a mining pool to win a block reward becomes significantly large. Then, an individual miner can secure its small but stable share of reward according to the ratio between its hash rate and the hash rate of the pool.

In this paper, we study the problem of mining pool selection in a PoW-based blockchain network, where each mining pool may adopt a different arbitrary block mining strategy~\cite{fisch2017socially}. By assuming that the individual miners are rational and profit-driven, we propose a model based on the evolutionary game to mathematically describe the dynamic mining-pool selection process in a large population of individual miners. Considering the computation power and propagation delay as the two major factors to determine the results of mining competitions, we focus on how these two factors as well as the cost of the computation resource (mainly in electricity) impact the strategies of the individual miners for pool selection. Based on a case study of two mining pools, we provide the theoretical analysis of evolutionary stability for the pool-selection dynamics. Our numerical simulation results provide the evidences that support our theoretical discoveries and further present the experimental insight into the impact of the arbitrary strategies on the reward outcomes of different mining pools.

\section{Problem formulation}\label{Sec:first}

\subsection{Financially Incentivized Block Mining with Proof-of-Work}
\label{lab_subsec_mining}
We consider a blockchain network adopting the Nakamoto consensus based on Proof-of-Work (PoW)~\cite{nakamoto2008bitcoin}. Assume that the network is composed of a large population of $N$ individual miners. For each miner, the chance of mining a new block is in proportion to the ratio between its individual hash rate for solving the crypto-puzzles in PoW and the total hash rate in the network. According to the Nakamoto consensus protocol, the miner of each confirmed block receives a fixed amount of blockchain tokens from the new block's coinbase and a flexible amount of transaction fees as the reward for maintaining the blockchain's consensus and approving the transactions~\cite{7423672}. We consider that the individual miners organize themselves into a set of $M$ mining pools, namely, $\mathcal{M}\!=\!\{1, 2, \ldots, M\}$. We further consider that each mining pool may set different requirement on the hash rate contributed by an individual miner trying to join the pool. Let $\omega_i$ denote the individual hash rate required by pool $i$ ($i\in\mathcal{M}$), and $x_i$ denote the miners' population fraction in pool $i$. Then, the probability for pool $i$ to mine a block in one consensus round can be expressed as:
\begin{equation}
  \label{eq_prob_win}
{\Pr}^{\textrm{mine}}_i(\mathbf{x},\pmb\omega)=\frac{\omega_i x_i}{\sum_{j=1}^{M}\omega_j x_j},
\end{equation}
where $\pmb\omega=[\omega_1,\ldots,\omega_M]^{\top}$, $\mathbf{x}=[x_1,\ldots,x_M]^{\top}$, $\sum_{i\in \mathcal{M}} x_i=1$ and $\forall i$, $x_i\ge0$.

After successfully mining a block, pool $i$ broadcasts the mined block to its neighbors in the hope that it will be propagated to the entire network and confirmed as the new head block of the blockchain. However, in the situation where more than one mining pool discover a new block at the same time, only the block that is first disseminated to the network will be confirmed by the network. All of the rest candidate blocks will be discarded (orphaned). According to the empirical studies in~\cite{Rizun2015, Houy2016}, the block propagation time is mainly determined by two factors, namely, the transmission delay over each link and the transaction verification time at each relaying node. For a block of size $s$, the transmission delay can be modeled as $\tau_p(s)\!=\!\frac{s}{\gamma c}$~\cite{Rizun2015}, where $\gamma$ is a network scale-related parameter, and $c$ is the average effective channel capacity of each link. On the other hand, since verifying a transaction requires roughly the same amount of computation, the block verification time can be modeled as a linear function  $\tau_v(s)\!=\!bs$, where $b$ is a parameter determined by both the network scale and the average verification speed of each node. Then, the average propagation time for a block of size $s$ can be expressed as:
\begin{equation}
  \label{eq_pro_time}
\tau(s)=\tau_p(s)+\tau_v(s)=\frac{s}{\gamma c}+bs.
\end{equation}

Based on (\ref{eq_pro_time}), the incidence of abandoned (i.e., orphaning) a valid block due to the propagation delay can be modeled as a Poisson process with mean $1/T$, where $T$ is maintained by the network as a fixed average mining time (e.g., 600s in Bitcoin)~\cite{Rizun2015}. Then, the probability of orphaning a valid candidate block of size $s$ in one consensus round is:
\begin{equation}
  \label{eq_prob_orphane}
  {\Pr}^{\textrm{orphan}}(s)=1-e^{-\tau(s)/T}=1-e^{-(\frac{s}{\gamma c}+bs) /T}.
\end{equation}
From (\ref{eq_prob_win}) and (\ref{eq_prob_orphane}), the probability for pool $i$ to ultimately win a block mining race with a block of size $s_i$ can be derived as:
\begin{equation}
  \label{eq_prob_winning}
{\Pr}^{\textrm{win}}_{i}(\mathbf{x}, \pmb{\omega}, s_i)=\frac{\omega_i x_i}{\sum_{j=1}^{M}\omega_j x_j}e^{-(\frac{s_i}{\gamma c}+bs_i)/T}.
\end{equation}

We assume that the transactions in the blockchain network are issued with an invariant rate of transaction fees. When the transactions are of fixed size, pool $i$'s mining reward from transaction fee collection can also be modeled as a linear function of the block size $s_i$. Let $\rho$ denote the price of transaction in a unit block size~\cite{Houy2016}. Then, the reward of pool $i$ from transaction fees can be written as $\rho s_i$. Let $R$ denote the fixed reward from the new block's coinbase. Then, the expected reward for pool $i$ can be expressed as follows:
\begin{equation}
  \label{eq_expected_reward}
  E\{r_i(\mathbf{x}, \pmb{\omega}, s_i)\}=(R+\rho s_i)\frac{\omega_i x_i}{\sum_{j=1}^{M}\omega_j x_j}e^{-(\frac{s_i}{\gamma c}+bs_i)/T}.
\end{equation}

Since the process of crypto-puzzle solving in PoW is computationally intensive, the rational miners also have to consider the cost of power consumption due to hash computation in the block mining process. Noting that the new blocks are discovered with a roughly fixed time interval, we denote the energy price for generating a unit hash query rate during that time interval by $p$. Then, we can obtain the expected payoff for an individual miner in pool $i$ as follows:
\begin{equation}
  \label{eq_payoff_individual}
  y_{i}(\mathbf{x}, \pmb{\omega}, s_i)=\frac{R+\rho s_i}{N x_i}\frac{\omega_i x_i}{\sum_{j=1}^{M}\omega_j x_j}e^{-(\frac{s_i}{\gamma c}+bs_i)/T}-p\omega_i.
\end{equation}\vspace*{-2mm}

\subsection{Mining Pool Selection as an Evolutionary Game}
Consider that the individual miners are rational and aim to maximize their net payoff given in (\ref{eq_payoff_individual}). Then, it is nature to model the process of mining pool selection in the population of individual miners as an evolutionary game. Mathematically, we can define the evolutionary game for mining pool selection as a 4-tuple: $\mathcal{G}=\langle\mathcal{N}, \mathcal{M}, \mathbf{x}, \{y_i(\mathbf{x}; \pmb{\omega}, s_i)\}_{i\in\mathcal{M}} \rangle$, where
\begin{itemize}
  \item $\mathcal{N}$ is the population of individual miners, $\vert\mathcal{N}\vert\!=\!N$.
  \item $\mathcal{M}\!=\!\{1, 2, \ldots, M\}$ is the set of mining pools, and $(w_i, s_i)$ is the mining strategy preference of each pool $i\in\mathcal{M}$.
  \item $\mathbf{x}=[x_1,\ldots,x_M]^{\top}\in\mathcal{X}$ is the vector of the population states, where $x_i$ represents the fraction of population that choose mining pool $i$. $\mathcal{X}=\{\mathbf{x}\in\mathbb{R}^M_{+}: \sum_{i\in \mathcal{M}} x_i=1\}$.
  \item $\{y_i(\mathbf{x}; \pmb{\omega}, s_i)\}_{i\in\mathcal{M}}$ is the set of individual miner's payoff in each mining pool. $y_i(\mathbf{x}; \pmb{\omega}, s_i)$ is given by (\ref{eq_payoff_individual}).
\end{itemize}

We note that $\omega_i$ and $s_i$ form the predetermined mining strategy of pool $i$. Given a population state $\mathbf{x}\in\mathcal{X}$, we can derive the average payoff of the individual miner in $\mathcal{N}$ based on (\ref{eq_payoff_individual}) as follows:
\begin{equation}
  \label{eq_average_payoff}
  \overline{y}(\mathbf{x})=\sum_{i=1}^{M}y_i(\mathbf{x}; \pmb{\omega}, s_i) x_i.
\end{equation}
Then, by the pairwise proportional imitation protocol~\cite{Josef2003}, the replicator dynamics for the evolution of the population states can be expressed by the following system of Ordinary Differential Equations (ODEs) $\forall i\in\mathcal{M}$~\cite{Josef2003}:
\begin{equation}
  \label{eq_rd}
\dot{x}_{i}(t)=f_i(\mathbf{x}(t); \pmb{\omega}, s_i)=x_{i}(t)(y_{i}(\mathbf{x}(t); \pmb{\omega}, s_i)-\overline{y}(\mathbf{x}(t))),
\end{equation}
where $\dot{x}_{i}(t)$ represents the growth rate of the size of pool $i$ with respect to time $t$.

We are interested in the Nash Equilibria (NE) of game $\mathcal{G}$ described by (\ref{eq_rd}). Let $Y(\mathbf{x})$ denote the vector of individual payoffs for all the mining pools, $Y(\mathbf{x})\!=\![y_1(\mathbf{x}),\ldots,y_M(\mathbf{x})]^{\top}$ and let $\mathcal{E}(Y)$ denote the set of NE in game $\mathcal{G}$. Then, $\mathcal{E}(Y)$ can be defined as follows~\cite{HOFBAUER20091665}:
\begin{Definition}[NE]
\label{def_NE}
 A population state $\mathbf{x^*}\!\in\!\mathcal{X}$ is an NE of the evolutionary game $\mathcal{G}$, i.e., $\mathbf{x}^*\!\in\!\mathcal{E}(Y)$, if for all feasible population state $\mathbf{x}\!\in\!\mathcal{X}$ the following inequality holds
\begin{equation}
  \label{eq_variational_inequality}
  (\mathbf{x}-\mathbf{x}^*)^{\top}Y(\mathbf{x}^*)\le0.
\end{equation}
\end{Definition}

It is straightforward that an NE is a fixed point of the replicator dynamics given by (\ref{eq_rd}), namely, $\forall i\!\in\!\mathcal{M}, f_{i}(\mathbf{x}(t); \pmb{\omega}, s_i)= 0$~\cite{Josef2003}. Then, we need to further investigate the stability of an NE state $\mathbf{x}^*\!\in\!\mathcal{E}(Y)$ for pool selection. Suppose that there exists another population state $\mathbf{x}'$ trying to invade state $\mathbf{x}^*$ by attracting a small share $\epsilon\!\in\!(0,1)$ in the population of miners to switch to $\mathbf{x}'$. Then, $\mathbf{x}'$ is an Evolutionary Stable Strategy (ESS) if the following condition holds for all $\epsilon\in(0, \overline{\epsilon})$:
\begin{equation}
  \label{eq_mutant}
  \sum_{i\in\mathcal{M}}x^*_iy_i((1-\epsilon)\mathbf{x}^*+\epsilon\mathbf{x}')\ge\sum_{i\in\mathcal{M}}x'_iy_i((1-\epsilon)\mathbf{x}^*+\epsilon\mathbf{x}').
\end{equation}
Based on (\ref{eq_mutant}), we can formally define the ESS as follows.
\begin{Definition}[ESS~\cite{HOFBAUER20091665}]
  \label{def_ESS}
  A population state $\mathbf{x^*}$ is an ESS of game $\mathcal{G}$, if there exists a neighborhood $\mathcal{B}\in\mathcal{X}$, such that $\forall \mathbf{x}\in\mathcal{B}-\mathbf{x}^*$, the condition $(\mathbf{x}-\mathbf{x}^*)^{\top}Y(\mathbf{x}^*)=0$ implies that
  \begin{equation}
    \label{eq_ess}
    (\mathbf{x}^*-\mathbf{x})^{\top}Y(\mathbf{x})\ge0.
  \end{equation}
\end{Definition}

In Algorithm~\ref{alg_rd}, we describe the strategy evolution of the $N$ individual miners following the revision protocol of pairwise proportional imitation~\cite{Weibull1997}. When receiving a signal for strategy revision of choosing a new pool, an individual miner switches from it current pool to the new pool probabilistically according to (\ref{eq_imitation_prob}). As the population size increases, the pairwise proportional imitation will asymptotically lead to the replicator dynamics described by the ODEs in (\ref{eq_rd}).
\begin{algorithm}[t!]
  \begin{algorithmic}[1]
  \STATE \textbf{Initialization}: $\forall i\in\mathcal{N}$, miner $i$ randomly selects a mining pool to start with.
  \STATE $t \gets 1$
  \WHILE{$\mathbf{x}$ has not converged \AND {$t<\textrm{MAX\_COUNT}$}}
    \FOR {$i\in\mathcal{N}$}
      \STATE $j\gets\textrm{Rand}(1, M)$ \COMMENT{Randomly selects a mining pool $j\in\mathcal{M}$}
      \STATE Determine whether to switch to pool $j$ according to the following probability of pool switching $\rho_{i,j}$:
      \begin{equation}
        \label{eq_imitation_prob}
        \rho_{i,j}=x_j\max(y_j(\mathbf{x}; \pmb{\omega}, s_j) - y_i(\mathbf{x}; \pmb{\omega}, s_i), 0).
      \end{equation}
    \ENDFOR
    \STATE $t \gets t+1$
  \ENDWHILE
  \end{algorithmic}
 \caption{Mining Pool Selection Following the Pairwise Proportional Imitation Protocol.}
 \label{alg_rd}
\end{algorithm}

\subsection{A Case Study of Two Mining Pools}
In this section, we consider a special case of a blockchain network with two mining pools, i.e., $M=2$. Let the population fraction of each pool be $x_1=x$, and $x_2=1-x$. From Definition~\ref{def_NE} and by solving $\dot{x}_i(t)=0, i\in [1,2]$, we can obtain Theorem~\ref{theorem_equilibrium} as follows.
\begin{theorem}
  \label{theorem_equilibrium}
  Based on the replicator dynamics in (\ref{eq_rd}), a blockchain network of two mining pools has three rest points in the form of $(x^*,1-x^*)$ with
\begin{equation}
  \label{equilibrium}
x^*\in\left\{0, 1, \frac{a-b}{Np(\omega_1-\omega_2)^2}-\frac{\omega_2}{\omega_1-\omega_2}\right\},
\end{equation}
where $a=(R+\rho s_1)\omega_1 e^{-(\frac{s_1}{\gamma c}+bs_1)/T}$, $b=(R+\rho s_2)\omega_2 e^{-(\frac{s_2}{\gamma c}+bs_2)/T}$ and $0<\frac{a-b}{Np(\omega_1-\omega_2)^2}-\frac{\omega_2}{\omega_1-\omega_2}<1$.
\end{theorem}
\begin{proof}
From $f_{i}(\mathbf{x}(t))= 0, \forall i\in \{1,2\}$, we have
\begin{align}
&f_{i}(\mathbf{x}(t))=x_{i}(t)(y_{i}(\mathbf{x}(t),a_i)-\overline{y}(\mathbf{x}(t)))\nonumber\\
&=x_{i}(t)(1\!-\!x_{i}(t))\left(\frac{a-b}{N(\omega_1x_{i}(t)\!+\!\omega_2(1\!-\!x_{i}(t)))}\!-\!p(\omega_1\!-\!\omega_2)\right).
\end{align}
Then, by solving $f_{i}(\mathbf{x}(t))=0$, we can obtain the three rest points for the case of two mining pools as $(x^*,1-x^*)$, where $x^*\in\left\{0, 1, \frac{a-b}{Np(\omega_1-\omega_2)^2} -\frac{\omega_2}{\omega_1-\omega_2}\right\}$. Since from any initial state $\mathbf{x}(0)\in\mathcal{X}$, the rest point of (\ref{eq_rd}) should stay in the interior of $\mathcal{X}$, we have the following condition:
\begin{equation}
  \label{eq_well_defined_rest_point}
  0<\frac{a-b}{Np(\omega_1-\omega_2)^2}-\frac{\omega_2}{\omega_1-\omega_2}<1.
\end{equation}
Then, the proof of Theorem~\ref{theorem_equilibrium} is completed.
\end{proof}

Now, we are ready to investigate the evolutionary stability of the three fixed points. In the case of $x^*=0$ and $x^*=1$, the population state is $(0,1)$ and $(1,0)$, respectively. We know that the two fixed points are of the similar form, since the individual payoff functions are similar for each mining pool. Therefore, we only need to check the case with $x_1=x^*=0$.
\begin{lemma}
  \label{lemma_ESS}
  For game $\mathcal{G}$ with two mining pools, 1) The rest point with $x^*=0$ is an ESS, if the conditions given by (\ref{condition11}) and (\ref{condition12}) hold:
\begin{equation}
\label{condition11}
\frac{a-b}{N\omega_2}-p(\omega_1-\omega_2)<0,
\end{equation}
\begin{equation}
  \label{condition12}
\left(\frac{a-b}{N\omega_2}-p(\omega_1-\omega_2)\right)\left(p\omega_2-\frac{b}{N\omega_2}\right)>0.
\end{equation}
2) If the conditions given by (\ref{condition1}) and (\ref{condition2}) hold, the rest point with $x^*=\frac{a-b}{Np(\omega_1-\omega_2)^2}-\frac{\omega_2}{\omega_1-\omega_2}$ is an ESS.
\begin{equation}
\label{condition1}
\frac{c(a(\omega_1+\omega_2)+\omega_1(-2b+Np\omega_1(\omega_2-\omega_1)))}{(a-b)}<0,
\end{equation}
\begin{equation}
  \label{condition2}
\frac{pc(-b\omega_1+a\omega_2)(a-b+Np\omega_1(\omega_2-\omega_1))}{(\omega_1-\omega_2)}>0,
\end{equation}
where $c=a-b+Np\omega_2(\omega_2-\omega_1)$.
\end{lemma}
\begin{proof}
According to Definition 2.6 in~\cite{Weibull1997}, the asymptotically stable state of the ODE system given in (\ref{eq_rd}) is guaranteed to be an ESS. When the replicator
dynamics is continuous-time, it is asymptotically stable if the Jacobian matrix of the dynamic system at the equilibrium is negative definite, or equivalently, if all the eigenvalues of the Jacobian matrix have negative real parts~\cite{Cressman2003}. For the replicate dynamic system given in (\ref{eq_rd}), the Jacobian matrix of the replicator dynamics in a two-mining-pool network is
\begin{equation}
  \label{eq_jacobian}
J=\left[
  \begin{array}{cc}
    \displaystyle\frac{\partial f_{1}(\mathbf{x})}{\partial x_{1}} & \displaystyle\frac{\partial f_{1}(\mathbf{x})}{\partial x_{2}} \\
    \displaystyle\frac{\partial f_{2}(\mathbf{x})}{\partial x_{1}} & \displaystyle\frac{\partial f_{2}(\mathbf{x})}{\partial x_{2}}
  \end{array}
\right]\Bigg|_{(x_1=x^*, x_2=1-x^*)}.
\end{equation}
Further, the elements in (\ref{eq_jacobian}) can be derived as follows:
\begin{align}
  \label{eq_11}
&\frac{\partial f_{1}(\mathbf{x})}{\partial x_{1}}
\!=\!(1\!-\!2x_1)\Bigg(\frac{a}{N(\omega_1x_1\!+\!\omega_2x_2)}\!-\!p\omega_1\Bigg)\nonumber\\
&-\frac{a\omega_1(x_1\!-\!x_1^2)}{N(\omega_1x_1\!+\!\omega_2x_2)^2}
\!-\!\frac{b\omega_2x_2^2}{N(\omega_1x_1\!+\!\omega_2x_2)^2}+p\omega_2x_2,
\end{align}

\begin{align}
  \label{eq_12}
&\frac{\partial f_{1}(\mathbf{x})}{\partial x_{2}}
\!=\!x_1\Bigg(p\omega_2\!-\!\frac{a\omega_2(1\!-\!x_1)}{N(\omega_1x_1\!+\!\omega_2x_2)^2}\nonumber\\
&+\frac{b\omega_2x_2}{N(\omega_1x_1\!+\!\omega_2x_2)^2}
\!-\!\frac{b}{N(\omega_1x_1\!+\!\omega_2x_2)}\Bigg),
\end{align}

\begin{align}
  \label{eq_21}
&\frac{\partial f_{2}(\mathbf{x})}{\partial x_{1}}
\!=\!x_2\Bigg(p\omega_1\!+\!\frac{a\omega_1x1}{N(\omega_1x_1\!+\!\omega_2x_2)^2}\nonumber\\
&-\frac{b\omega_1(1\!-\!x_2)}{N(\omega_1x_1\!+\!\omega_2x_2)^2}
-\frac{a}{N(\omega_1x_1\!+\!\omega_2x_2)}\Bigg),
\end{align}

\begin{align}
\label{eq_22}
&\frac{\partial f_{2}(\mathbf{x})}{\partial x_{2}}
\!=\!(1-\!2x_2)\Bigg(\frac{b}{N(\omega_1x_1\!+\!\omega_2x_2)}\!-\!p\omega_2\Bigg)\nonumber\\
&-\frac{b\omega_2(x_2\!-\!x_2^2)}{N(\omega_1x_1\!+\!\omega_2x_2)^2}
\!-\!\frac{a\omega_1x_1^2}{N(\omega_1x_1\!+\!\omega_2x_2)^2}+p\omega_1x_1.
\end{align}

Based on (\ref{eq_11})-(\ref{eq_22}), we have
\begin{itemize}
  \item [1)] After some tedious mathematical manipulations, the determinants of the principle minors of $J$ at $x^*=0$ should satisfy the following conditions to guarantee the negative definiteness of $J$:
\begin{equation}
\label{condition21}
\det(J_{11})
\!=\!\frac{a-b}{N\omega_2}\!-\!p(\omega_1\!-\!\omega_2)<0,
\end{equation}
\begin{equation}
\label{condition22}
\det\left(J\right)\!=\!\left(\frac{a\!-\!b}{N\omega_2}\!-\!p(\omega_1\!-\!\omega_2)\right)(p\omega_2\!-\!\frac{b}{N\omega_2})>0.
\end{equation}

\item [2)] Similarly, at $x^*\!=\!\frac{a-b}{Np(\omega_1-\omega_2)^2}\!-\!\frac{\omega_2}{\omega_1-\omega_2}$, the following conditions can be obtained for the negative definiteness of $J$ after some mathematical manipulations:
\begin{align}
\label{condition23}
\det(J_{11})
&\!=\!\frac{c\left(a(\omega_1+\omega_2)\!+\!\omega_1(-2b\!+\!Np\omega_1(\omega_2\!-\!\omega_1))\right)}{N(a\!-\!b)(\omega_1\!-\!\omega_2)^2}\nonumber\\
&<0,
\end{align}
\begin{equation}
\label{condition24}
\det(J)\!=\!\frac{pc(-b\omega_1\!+\!a\omega_2)(a\!-\!b\!+\!Np\omega_1(\omega_2\!-\!\omega_1))}{N(a-b)^2(\omega_1-\omega_2)}>0.
\end{equation}
\end{itemize}
Then, the proof to Lemma 1 is completed.
\end{proof}

We note that the blockchain network is comprised by a large population of individual miners in the real-world scenarios. Then, from Lemma~\ref{lemma_ESS}, we can employ the asymptotic analysis and obtain the following theorem on evolutionary stability of the rest points.
\begin{theorem}
  \label{theorem_stable} Assume that the population size $N$ is sufficiently large. Then, neither of the rest points with $x^*\in\{0,1\}$ is evolutionary stable. The rest point with $x^*\!=\!\frac{a-b}{Np(\omega_1-\omega_2)^2}\!-\!\frac{\omega_2}{\omega_1-\omega_2}$ is an ESS if the following conditions are satisfied:
\begin{align}
\label{conditionESS}
  \left\{
  \begin{array}{ll}
    a-b<0,\\
    (b\omega_1-a\omega_2)(\omega_2-\omega_1)>0.
  \end{array}\right.
\end{align}
\end{theorem}
\begin{proof}
First, at the rest point with $x^*=0$, by Lemma \ref{lemma_ESS}, we can obtain the following conditions for the Jacobian if $\omega_1\leq\omega_2$,
\begin{equation}
\lim_{N\rightarrow+\infty}\det(J_{11})\!=\!\lim_{N\rightarrow+\infty}\frac{a-b}{N\omega_2}-p(\omega_1-\omega_2)\geq0.
\end{equation}
Then, the Jacobian matrix is not negative definite. Alternatively, if $\omega_1>\omega_2$, we have
\begin{equation}
\lim_{N\rightarrow+\infty}\det(J_{11})\!=\!
\lim_{N\rightarrow+\infty}\frac{a-b}{N\omega_2}-p(\omega_1-\omega_2)<0,
\end{equation}
and
\begin{equation}
  \label{condition32}
\lim_{N\rightarrow+\infty}\det(J)\!=\!
\lim_{N\rightarrow+\infty}(\frac{a\!-\!b}{N\omega_2}\!-\!p(\omega_1\!-\!\omega_2))(p\omega_2\!-\!\frac{b}{N\omega_2})\!<\!0.
\end{equation}
Again, the Jacobian matrix cannot be negative definite. Then, the rest point with $x^*=0$ is not an ESS. Following the same procedure, we can show that the rest point with $x^*=1$ is not evolutionary stable, either.

By~\cite{Cressman2003}, we know that any rest point in the interior of $\mathcal{X}$ is an NE. Then, for the NE with $x^*\!=\!\frac{a-b}{Np(\omega_1-\omega_2)^2}\!-\!\frac{\omega_2}{\omega_1-\omega_2}$, following Lemma \ref{lemma_ESS}, we obtain
\begin{align}
  \label{proof_theorem_ess_1}
&\lim_{N\rightarrow+\infty}\det(J_{11})\nonumber\\
&\!=\!\lim_{N\rightarrow+\infty}\frac{(a-b+Np\omega_2(\omega_2-\omega_1))a(\omega_1+\omega_2)}{N(a-b)(\omega_1-\omega_2)^2}+
\nonumber\\
&\frac{(a-b+Np\omega_2(\omega_2-\omega_1))\omega_1(-2b+Np\omega_1(\omega_2-\omega_1))}{N(a-b)(\omega_1-\omega_2)^2}
\nonumber\\
&\!=\!\lim_{N\rightarrow+\infty}\frac{Np^2\omega_1\omega_2}{a-b},
\end{align}
and
\begin{align}
  \label{proof_theorem_ess_2}
&\lim_{N\rightarrow+\infty}\det\left(J\right)
\!=\!\lim_{N\rightarrow+\infty}\frac{p(a-b+Np\omega_2(\omega_2-\omega_1))}{N(a-b)^2(\omega_1-\omega_2)}\cdot
\nonumber\\
&\frac{(-b\omega_1+a\omega_2)(a-b+Np\omega_1(\omega_2-\omega_1))}{N(a-b)^2(\omega_1-\omega_2)}
\nonumber\\
&\!=\!\lim_{N\rightarrow+\infty}\frac{Np^3\omega_1\omega_2(b\omega_1-a\omega_2)(\omega_2-\omega_1)}{(a-b)^2}.
\end{align}
By (\ref{proof_theorem_ess_1}) and (\ref{proof_theorem_ess_2}), the Jacobian matrix is negative definite if the conditions given in (\ref{conditionESS}) are satisfied, hence the NE $(x^*,1-x^*)$ is an ESS. Then, the proof to Theorem 2 is completed.
\end{proof}

\section{Evolution Analysis}
\label{Sec:second}
In this section, we conduct several numerical simulations and provide the performance evaluation of the individual miners' pool-selection strategies in different situations. We first consider a blockchain network with $N=5000$ individual miners, which evolve to form two mining pools (i.e., $M=2$). For the purpose of demonstration, we set the block generation parameters as $\lambda=1/600$, $\frac{1}{\gamma c}+b=0.005$, $R=1000$, $\rho=2$ and $p=0.01$. We also set the initial population state as $\mathbf{x}=[0.75,0.25]$. We first consider that the two pools adopt their mining strategies with the same block size, $s_1=s_2=100$, and different computation power contribution, $\omega_1=30$ and $\omega_2=20$. By Theorem~\ref{theorem_stable}, we know that such strategy adaptation satisfies the condition for an ESS in the interior of the simplex $\mathcal{X}$. Figure~\ref{fig1_a} demonstrates the evolution of the miners' pool-selection strategies. According to Figure~\ref{fig1_b}, the strategies converge to a global ESS of $(0.4, 0.6)$, which is in accordance with our theoretical prediction. We also observe that relatively fewer miners choose to join the pool requiring a higher hash rate (i.e., pool 1) at the ESS.  This is because a higher computation power requirement will lead to an increase in the mining cost, which exceeds the profit improvement that the miner can obtain in that pool.
\begin{figure}[t]
\centering     %%% not \center
\subfigure[]{\label{fig1_a}\includegraphics[width=.245\textwidth]{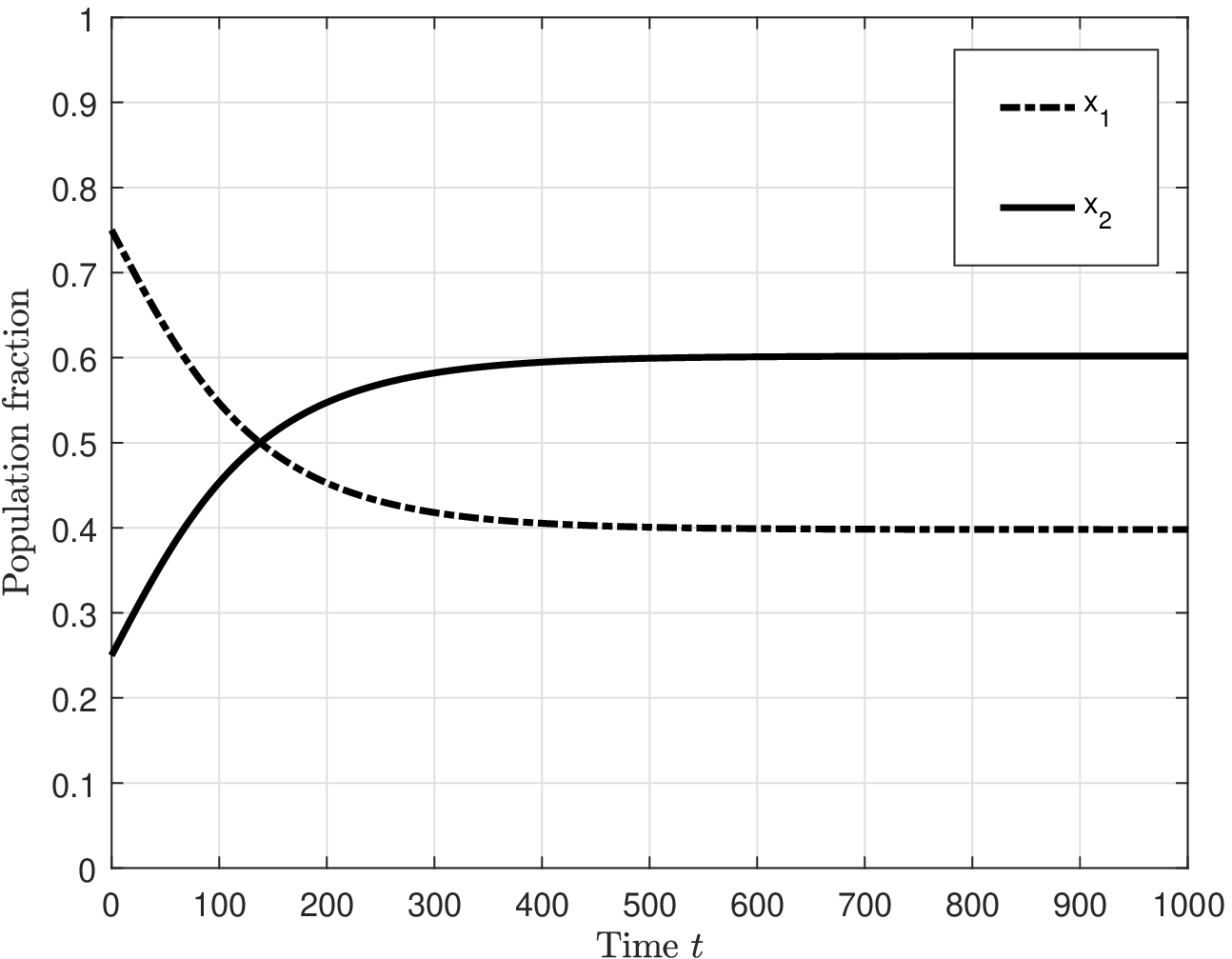}}
\subfigure[]{\label{fig1_b}\includegraphics[width=.237\textwidth]{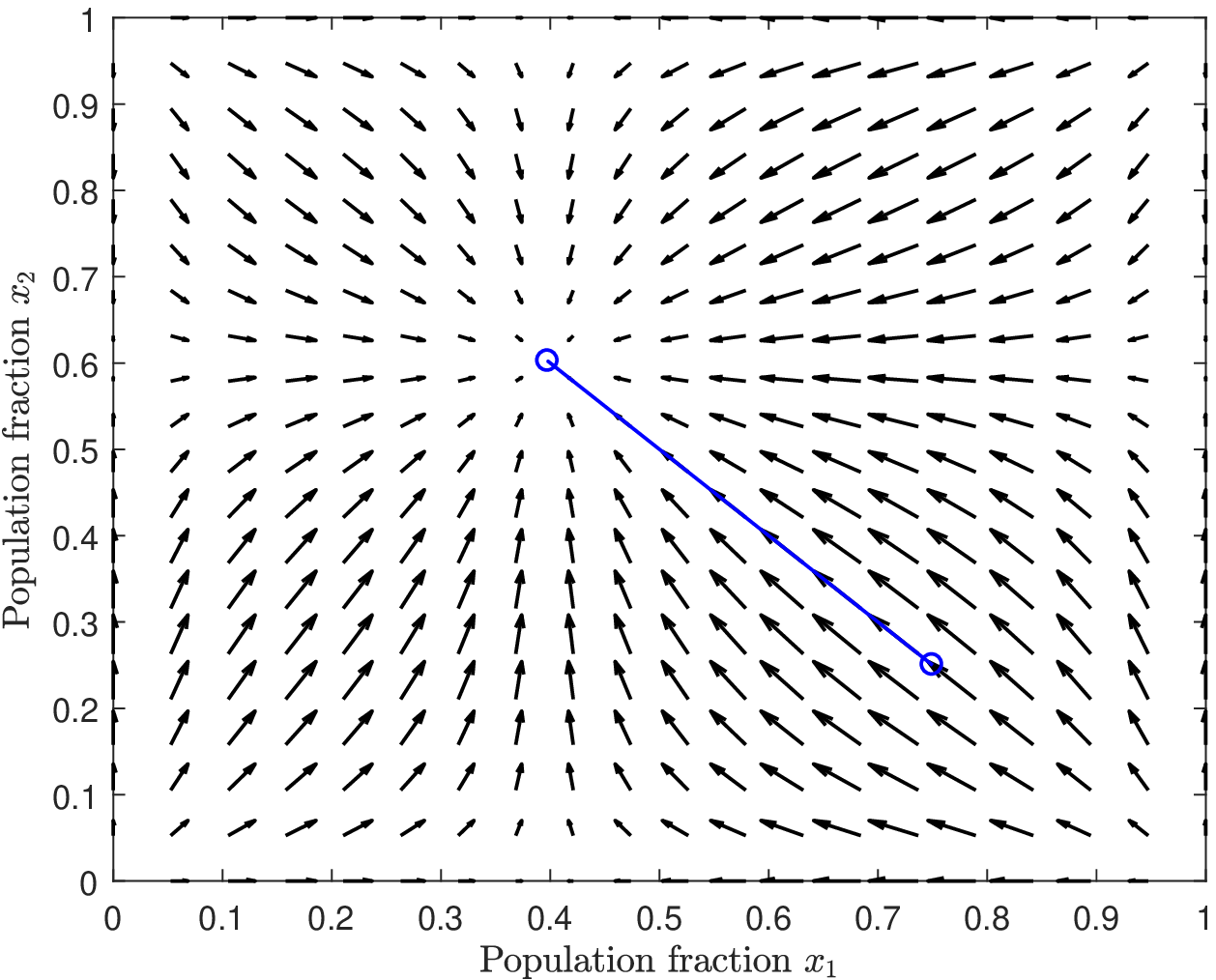}}
\caption{(a) Evolution of the miners' population states over time with two mining pools. (b) Replicator dynamics of the pool-selection strategies and the evolution trajectory starting from\vspace{-1mm} $\mathbf{x}(0)\!=\!(0.75, 0.25)$.}
\label{fig1}
\end{figure}

Further, we analyze the influence of the network condition on the pool-selection strategies of the individual miners. In Figure~\ref{fig3}, we show the evolution of the stable population states and the corresponding average payoff of the individual miners with respect to varied delay coefficient $\frac{1}{\gamma c}+b$. In the simulation, we adopt the same mining strategies as in Figure~\ref{fig1}. Figure~\ref{fig3_a} shows that as the propagation delay coefficient increases, more miners will tend to join the pool with a smaller hash rate requirement ($\omega_2\!=\!20$). Jointly considering the payoffs at NE shown in Figure~\ref{fig3_b}, we know that a larger delay coefficient leads to a higher probability of orphaning blocks of the same size. As a result, the miners prefer to join the pool that induces lower mining cost. We can also observe in Figure~\ref{fig3_b}  that the payoffs of the mining pool remain unchanged at zero. This phenomenon can be interpreted as a situation of market equilibrium where the demand for the hash rate exactly meets the supply with the current settings of reward parameters.
\begin{figure}[t]
\centering     %%% not \center
\subfigure[]{\label{fig3_a}\includegraphics[width=.2375\textwidth]{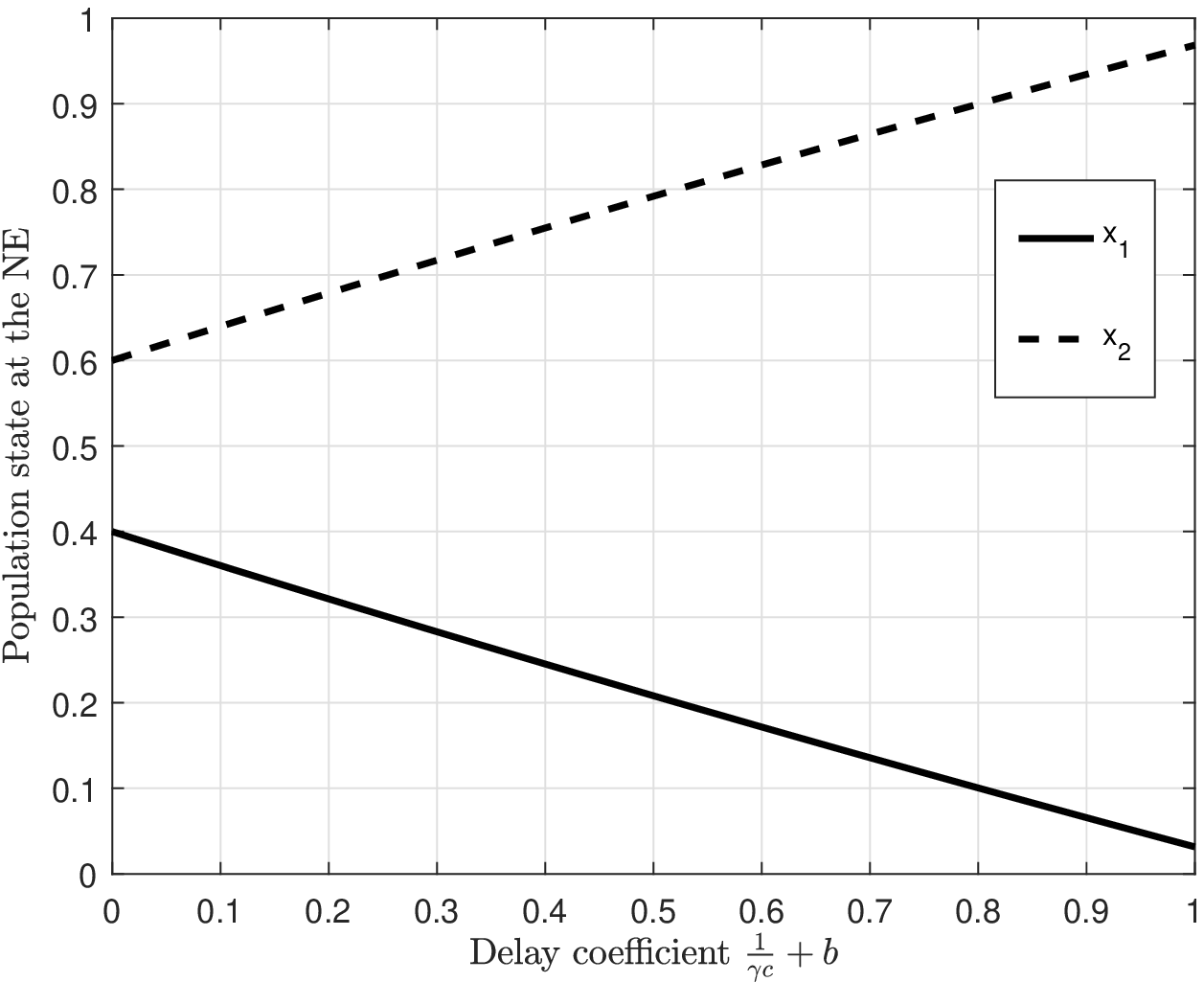}}
\subfigure[]{\label{fig3_b}\includegraphics[width=.244\textwidth]{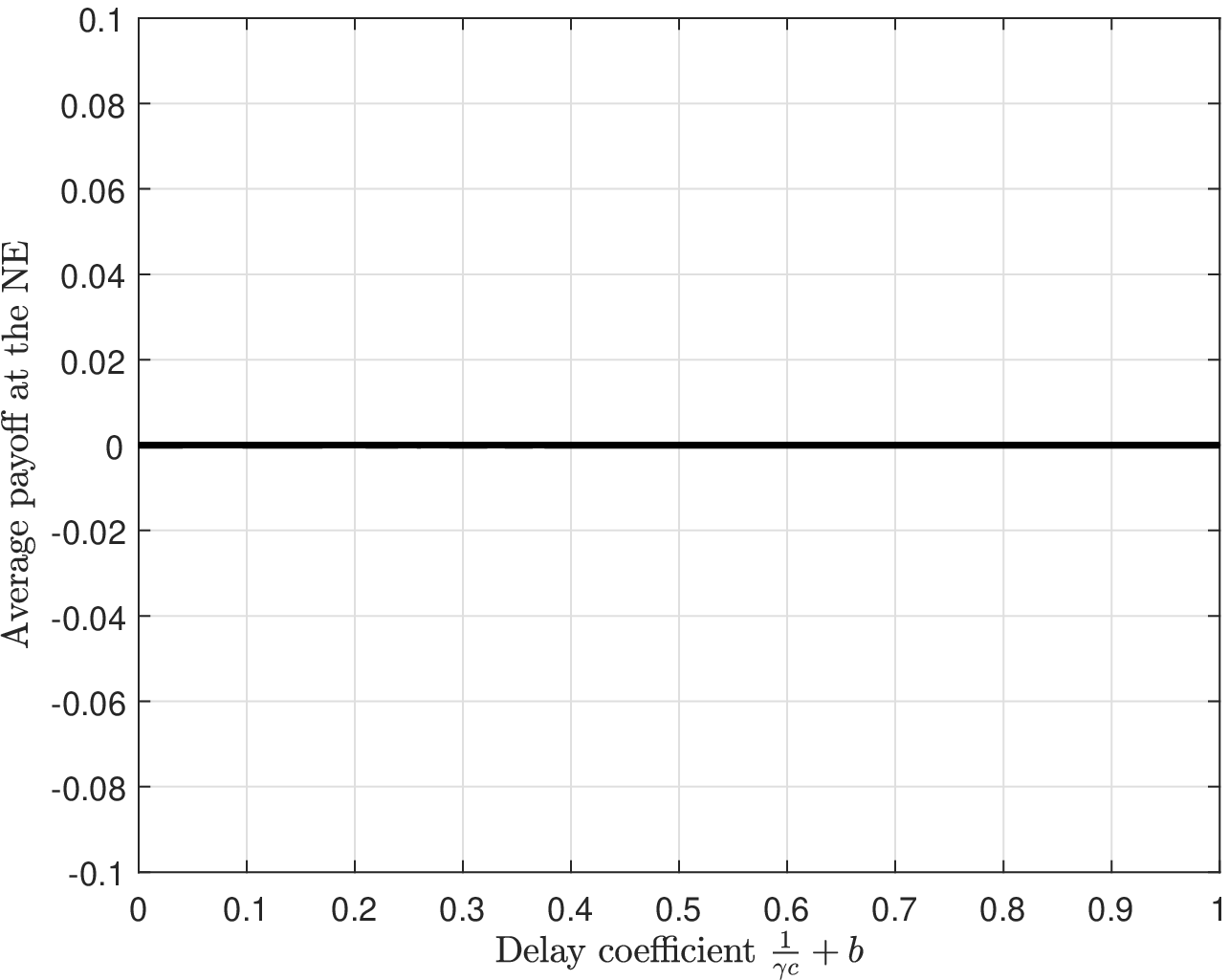}}
\caption{(a) Population state at the ESS vs. varying delay coefficient $\frac{1}{\gamma c}+b$. (b) Average payoff of the miners at the ESS vs. varying delay coefficient $\frac{1}{\gamma c}+b$.}
\label{fig3}
\end{figure}
\begin{figure}[t]
\centering     %%% not \center
\subfigure[]{\label{fig4_a}\includegraphics[width=.2375\textwidth]{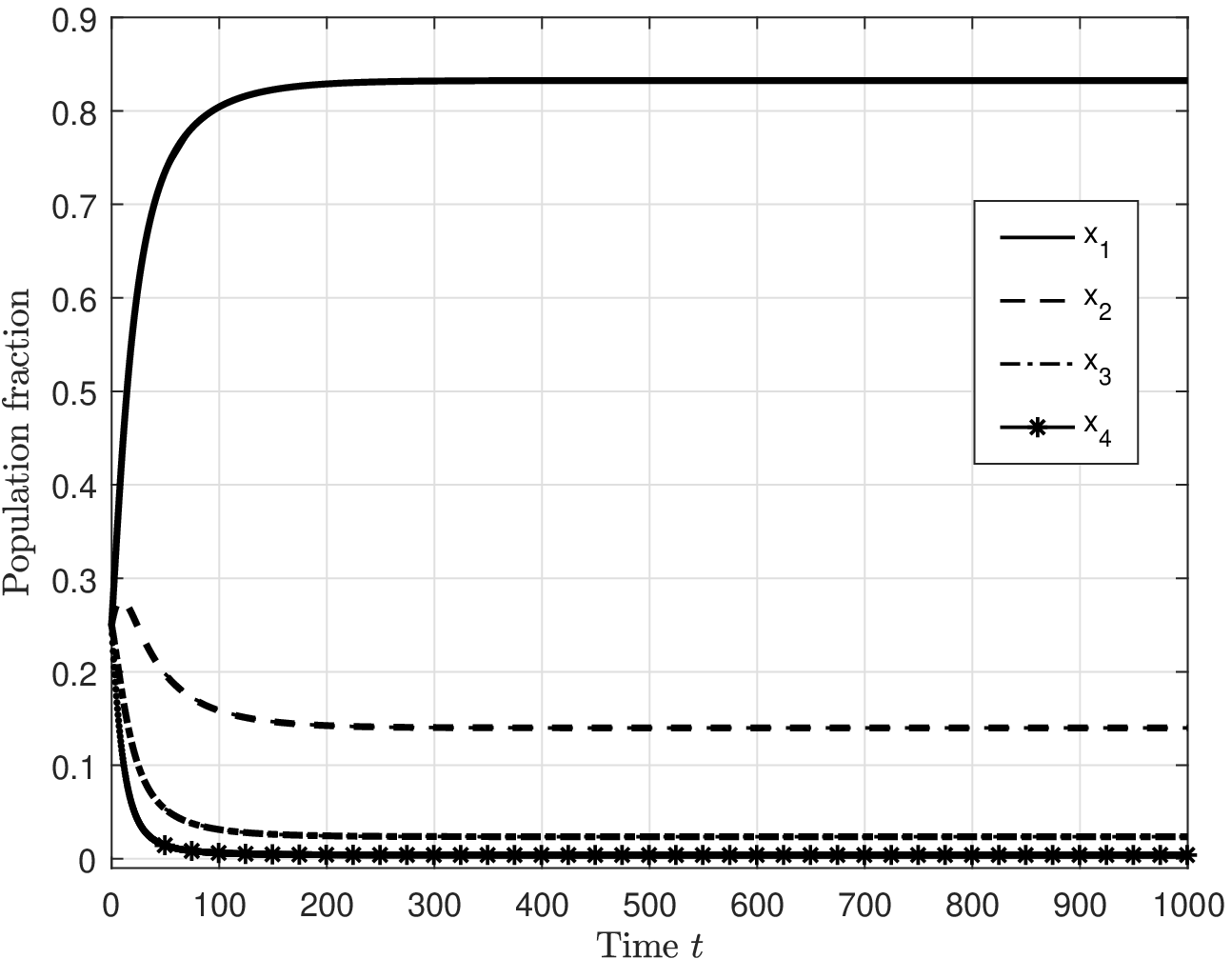}}
\subfigure[]{\label{fig4_b}\includegraphics[width=.243\textwidth]{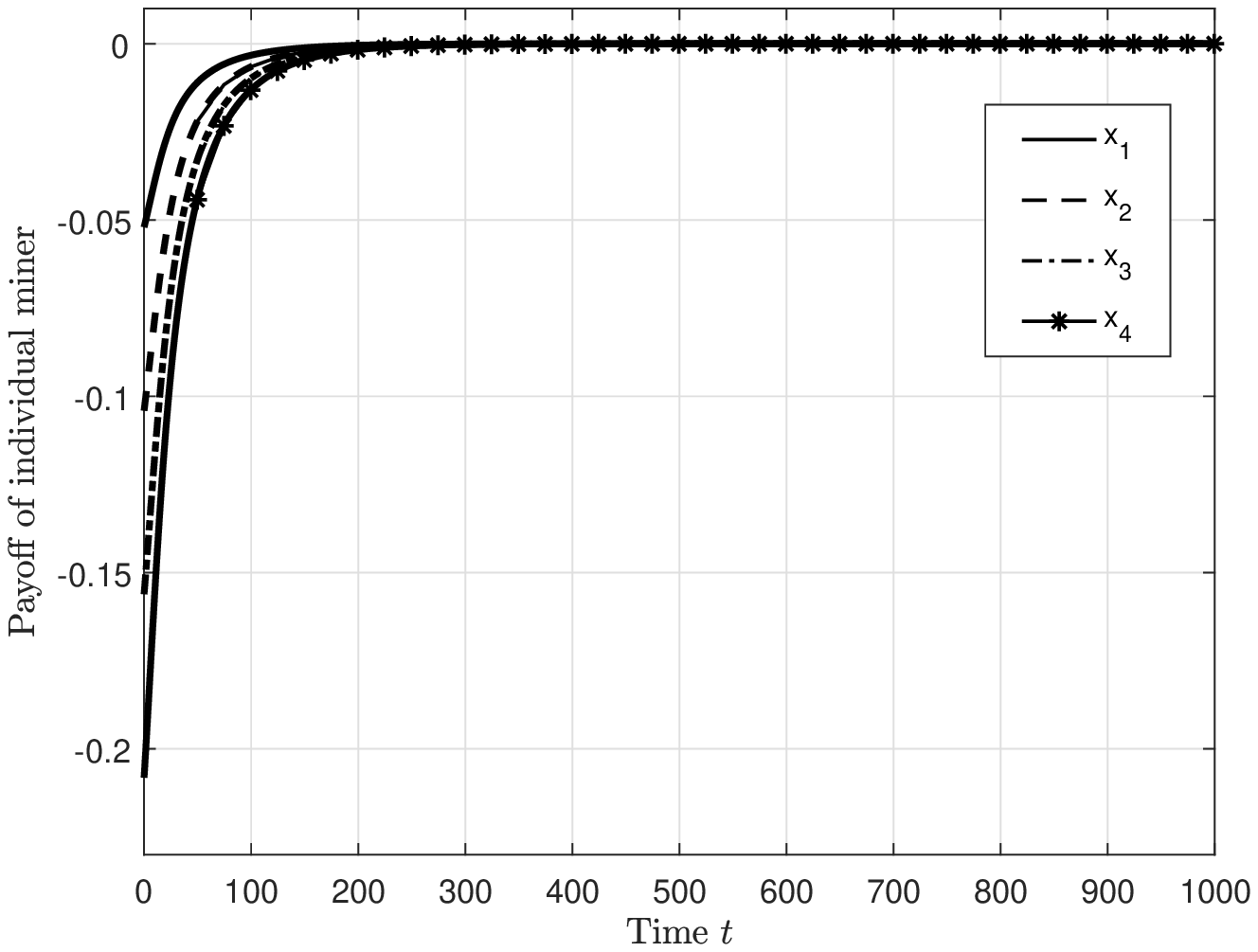}}
\caption{(a) The population states evolution with respect to different delay coefficient $\frac{1}{\gamma c}+b$, where mining strategy variables are $s_1=100$, $s_2=120$ and $\omega_1=20$, $\omega_2=20$. (b) Payoff evolution with respect to different delay coefficient $\frac{1}{\gamma c}+b$.}
\label{fig4}
\end{figure}

Finally, we consider a more general situation with four mining pools, where each pool adopts in their mining strategy the same block size as $s_i=100$ ($1\le i\le 4$) and different requirement on the hash rate contribution as $\omega_1\!=\!10$, $\omega_2\!=\!20$, $\omega_3\!=\!30$ and $\omega_4\!=\!40$. The evolution of the miner population states is presented in Figure~\ref{fig4_a}. In the considered case, we observe that when the miners' pool-selection strategies converge to the equilibrium, selecting pool 1 becomes a dominating strategy since by contributing a higher hash rate, the profit gain is unable to cover the power consumption cost for each miner. Then, the individual miners prefer to decrease their dedicated hash rate since they are sensitive to the mining cost. Figure~\ref{fig4_b} show that the payoffs of joining a pool evolves from negative value to zero. Again, this indicates a situation where the block mining business becomes a perfect competition market with an NE payoff of zero, and no miner can switch its pool selection without undermining some other miner's payoff at the equilibrium.

\section{Conclusion}\label{Sec:Conclusion}
In this paper, we have investigated the dynamic mining pool-selection problem in a blockchain network using Nakamoto consensus protocol. We model the dynamics of the individual miner's pool-selection strategies as an evolutionary game. In particular, we have considered the computation power and propagation delay as two major factors that determine the outcome of the block mining competition. Furthermore, we have theoretically analyzed the evolutionary stability of the pool selection dynamics based on a case study of two mining pools. For the case of two mining pools, we have shown that the blockchain network conditionally admits a unique evolutionary stable state. Our simulation results have provided the numerical evidence for our theoretical discoveries.

%--------------------------------------------------------------------------------------------------------------------------
\bibliography{reference}
\bibliographystyle{IEEEtran}

\end{document}